\documentclass[11pt]{article}
 
 \usepackage[letterpaper, margin=1in]{geometry}

\usepackage{amsmath,amssymb, amsthm, mathrsfs} 
\usepackage{xspace}
\usepackage{color}
\usepackage{subfig}
\usepackage{hyperref}
\usepackage{url}

\usepackage{wrapfig}
\usepackage{chngcntr}
\counterwithin{figure}{section}
\counterwithin{table}{section}
\usepackage{graphicx}
\DeclareGraphicsExtensions{.pdf}

\usepackage{thmtools}
\usepackage{thm-restate}

\usepackage{microtype}
\usepackage[T1]{fontenc}
\usepackage{charter}
\usepackage[expert]{mathdesign}

\usepackage{enumitem}
\setlist[description]{font=\normalfont\scshape}

\newtheorem{theorem}{Theorem}[section]
\newtheorem{definition}[theorem]{Definition}

\newtheorem{lemma}[theorem]{Lemma}

\newtheorem{proposition}[theorem]{Proposition}

\newcommand{\retheorem}{lemma}
\newcommand{\restatetheorem}[1]{\renewcommand{\retheorem}{rlemma}#1\renewcommand{\retheorem}{lemma}}

\renewcommand{\phi}{\varphi}

\newcommand{\calD}{\ensuremath{\mathcal{D}}\xspace}
\newcommand{\calH}{\ensuremath{\mathcal{H}}\xspace}

\newcommand{\matching}[1][\mathsf{V W}]{$#1$-matching\xspace}
\newcommand{\matchings}[1][\mathsf{V W}]{$#1$-matchings\xspace}

\newcommand{\free}[1]{\ensuremath{#1}-free\xspace}

\newcommand{\memory}[1]{\ensuremath{\mathfrak{M}_{#1}}\xspace}

\newcommand{\FPF}{flippable product-family\xspace}
\newcommand{\FPFs}{flippable product-families\xspace}

\newcommand{\WS}[1]{$#1$-winning strategy\xspace}
\newcommand{\WSs}[1]{$#1$-winning strategies\xspace}

\newcommand{\PCR}  {\ensuremath{\mathsf{PCR}}\xspace}
\newcommand{\RES}  {\ensuremath{\mathsf{RES}}\xspace}

\DeclareMathOperator{\Space}{MSpace}
\DeclareMathOperator{\dom}{dom}
\newcommand{\st}{\ :\ }

\newcommand{\Star}{\ensuremath{\lambda}\xspace}

\newcommand{\HH}{\ensuremath{\mathcal{H}}\xspace}
\newcommand{\R}{\ensuremath{\mathcal{R}}\xspace}

\newcommand{\xvec}[3]{\ensuremath{{#1}_{#2},\ldots,{#1}_{#3}}\xspace}

\newcommand{\CoverGame}[1]{%
\ensuremath{\mathsf{CoverGame}_{\mathsf{VW}}(#1)}\xspace%
}

\newcommand{\LL}{\ensuremath{\mathscr{L}}\xspace}
\newcommand{\FF}{\ensuremath{\mathscr{F}}\xspace}

\newcommand{\Choose}{{\sf Choose}\xspace}
\newcommand{\Cover}{{\sf Cover}\xspace}
\newcommand{\CovGam}{{\sf CoverGame}\xspace}

\newcommand{\matchingprop}{\matching property\xspace}
\newcommand{\expander}[1]{$(#1)$-bipartite expander\xspace}

\newcommand{\dfn}[1]{\textbf{\textit{#1}}\xspace}

\newcommand{\bfrac}[2]{\left(\frac{#1}{#2}\right)}
\newcommand{\sqbs}[1]{\left[ #1 \right]}

\begin{document}
\title{Space proof complexity for random $3$-CNFs}
\author{
Patrick Bennett%
\footnote{
Computer Science Department, University of Toronto, 10 Kings College Road, M5S 3G4 Toronto, Canada, \texttt{\{patrickb, molloy\}@cs.toronto.edu}.
}
\and
Ilario Bonacina%
\footnote{
Computer Science Department, Sapienza University of Rome, via Salaria 113, 00198 Rome, Italy, \texttt{\{bonacina, galesi, huynh, wollan\}@di.uniroma1.it}.
}
\and 
Nicola Galesi$^\dagger$
\and 
Tony Huynh$^\dagger$%
\footnote{Supported by the European Research Council under the European Union's Seventh Framework Programme (FP7/2007-2013)/ERC Grant Agreement no. 279558.}
\and 
Mike Molloy$^*$
\and 
Paul Wollan$^{\dagger \ddagger}$
}
\date{}
\maketitle
\thispagestyle{empty}

\begin{abstract}
 We investigate the space complexity of refuting  $3$-CNFs in Resolution and  algebraic systems.  
We prove  that every {\em Polynomial Calculus with Resolution} refutation of a random $3$-CNF $\phi$ in $n$ variables requires, with high probability, $\Omega(n)$  {\em distinct monomials} to be kept simultaneously in memory.  
The same construction also proves that every {\em Resolution} refutation $\phi$ requires, with high probability, $\Omega(n)$ clauses each of width $\Omega(n)$ to be kept at the same time in memory. This gives a $\Omega(n^2)$ lower bound for the {\em total space} needed in Resolution to refute $\phi$. These results are best possible (up to a constant factor) and answer questions about space complexity of $3$-CNFs posed in \cite{FilmusLNTR12,FilmusLMNV13,BonacinaGT14,bg15}.

The main technical innovation is a variant of  {\em Hall's Lemma}.   
We show that in bipartite graphs $G$  with bipartition $(L,R)$ and left-degree at most 3, $L$ can be covered by certain families of disjoint paths, called {\em \matchings}, provided that $L$ {\em expands} in $R$ by a factor of $(2-\epsilon)$, for $\epsilon < \frac{1}{23}$. 
\end{abstract}

\newpage
\setcounter{page}{1}
\section{Introduction}
During the last decade, an active line of research in proof
complexity has been the space complexity of proofs and how space
is related to other complexity measures (like size, length, width, degree) \cite{EstebanT01, AlekhnovichBRW02,Ben-SassonG03,Ben-Sasson02,AtseriasD08,Ben-SassonN08,Nordstrom09, BenSassonN11,FilmusLNTR12,FilmusLMNV13,BonacinaGT14,bg15}. 
This investigation has raised several important foundational questions.  Some of these have been solved, while several others are still open and challenging (see \cite{Nordstrom13} for a survey on this topic). 
Space of proofs concerns the minimal memory occupation of algorithms verifying the correctness of proofs in concrete propositional proof systems, and is thus also relevant in more applied algorithmic contexts.  
For instance, a major problem in state of the art SAT-solvers is memory consumption. In proof complexity, this resource is modeled by proof space.  
It is well-known that SAT-solvers used in practice (like CDCL) are based on low-level proof systems such as {\em Resolution}. 

In this work we focus on two well known proof systems that play a central role in proof complexity: {\em Resolution} \cite{Robinson:1965,Blake37} and {\em Polynomial Calculus} \cite{CleggEI96}. 
Resolution (\RES) is a refutational proof system for unsatisfiable propositional CNF formulas using only one logical rule: $ \frac{A \vee x \;\;\;\;\; \neg x \vee B}{A\vee B}$.  
Polynomial calculus is an algebraic refutational proof system for unsatisfiable sets of polynomials (over $\{0,1\}$ solutions) based on two rules: {\em linear combination} of polynomials and {\em multiplication} by variables.  In this article, we consider the stronger system {\em Polynomial Calculus with Resolution} (\PCR) which extends both Resolution and Polynomial Calculus \cite{AlekhnovichBRW02}.

Several different measures for proof space  were investigated for these two systems \cite{EstebanT01,AlekhnovichBRW02,Ben-Sasson02,AtseriasD08,Ben-SassonN08,Nordstrom09, BenSassonN11,FilmusLNTR12,BonacinaGT14,bg15}.  
In this work we focus on {\em total space} (for \RES), which is the maximum number of variables (counted with repetitions) to be kept simultaneously in  memory while verifying a proof; and {\em monomial space} (for \PCR), which is the maximum number of distinct monomials to be kept simultaneously  in  memory while verifying a proof.  
Both measures were introduced in  \cite{AlekhnovichBRW02}, where some preliminary lower and upper bounds were given. In particular, for every unsatisfiable CNF in $n$ variables, there is an easy upper bound
of $O(n)$ for monomial space in \PCR and $O(n^2)$ for total space in \RES.

Major open problems about these two measures were solved only recently in \cite{FilmusLMNV13, bg15,BonacinaGT14}.  
In particular, \cite{bg15,BonacinaGT14} prove that, for $r\geq 4$, random $r$-CNFs over $n$ variables require $\Theta(n^2)$ total space in resolution and $\Theta(n)$ monomial space in \PCR. 
However, it is not at all obvious how to generalize the techniques in  \cite{bg15,BonacinaGT14} to handle $3$-CNFs.  Indeed, it is an open problem whether there is any family of $3$-CNFs requiring large {\em total space} (in \RES) and {\em monomial space} (in \PCR).  In this work, we resolve this problem by proving that random $3$-CNFs also require $\Theta(n^2)$ {\em total space} (in \RES) and  $\Theta(n)$ {\em monomial space} (in \PCR).  

\paragraph{Results.}

Let $\phi$ be a random $3$-CNF in $n$ variables.  We prove that every \PCR  refutation of $\phi$  requires, with high probability, $\Omega(n)$  distinct monomials to be kept simultaneously in memory (Theorem \ref{cor:rand}).  
Moreover, every \RES refutation of $\phi$ has, with high probability, $\Omega(n)$ clauses each of width $\Omega(n)$ to be kept at the same time in memory (Theorem \ref{cor:rand}). This gives a $\Omega(n^2)$ lower bound for the total space of every \RES refutation of $\phi$.  These results resolve questions about space complexity of $3$-CNFs mentioned in \cite{FilmusLMNV13,BonacinaGT14,bg15,FilmusLNTR12}.

Both results follow using the framework proposed in \cite{bg15}, where the construction of suitable families of assignments called {\em \WSs{k}}  (Definition \ref{def:kex}) leads to monomial space lower bounds in \PCR (Theorem \ref{thm:lowerbound}).
This construction is made possible by a modification of Hall's Lemma \cite{Hall} for matchings to \matchings (Lemma \ref{lem:Hall}). 

\begin{definition}[\matching]\label{def:h-k-matchings}
Let $G$ be a bipartite graph with bipartition $(L,R)$. A \dfn{\matching} in $G$ is a subgraph $F$ of $G$ such that each connected component of $F$ is a path with at most $4$ edges and both endpoints in $R$.
A \matching $F$ {\em covers} a set of vertices $S$ if $S\subseteq V(F)$. Define $L(F)=V(F)\cap L$ and $R(F)=V(F)\cap R$.
\end{definition}

\begin{wrapfigure}[17]{l}{.2\textwidth}
\scriptsize
\centering
\includegraphics{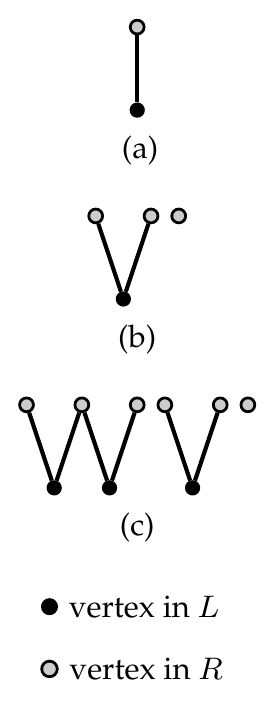}
\caption{%
}
\label{fig:matching}
\end{wrapfigure}

Figure \ref{fig:matching} compare {\em matchings} (Figure \ref{fig:matching}.(a)), {\em 2-matchings} as used in \cite{BonacinaGT14,bg15} (Figure \ref{fig:matching}.(b)) and \matchings (Figure \ref{fig:matching}.(c)). Note that for technical reasons, we allow 2-matchings and  \matchings to contain isolated vertices from $R$.
We can now state our variant of Hall's Lemma. 
This lemma and its proof are independent from the proof complexity results and might be useful in other contexts. 

\begin{restatable}[$(2-\epsilon)$-Hall's Lemma]{\retheorem}{restateHall}\label{lem:Hall}
Let $ \epsilon < \frac{1}{23}$.  Let $G$ be a bipartite graph with bipartition $(L,R)$ such that each vertex in $L$ has degree at most $3$ and no pair of degree $3$ vertices in $L$ have the same set of neighbors. If $|N_G(L)|\geq (2-\epsilon)|L|$, and each proper subset of $L$ can be covered by a \matching, then $L$ can be covered by a \matching. 
\end{restatable}

Note that the converse of Lemma \ref{lem:Hall} does not hold (unlike in Hall's Lemma).

\paragraph{Outline of the paper.}

Section \ref{sec:prelimin} contains some preliminary notions about proof complexity. In particular, the formal definitions of Resolution and Polynomial Calculus with Resolution, the model of space (based on \cite{EstebanT01, AlekhnovichBRW02}) and the formal definition of total space and monomial space. We present a simplified (but less general) version of the \WSs{k} of \cite{bg15} (Definition \ref{def:kex}). 
These \WSs{k} were used in \cite{bg15} to prove monomial space lower bounds for \PCR. Here we use the  same \WSs{k} also to prove  total space lower bounds for \RES. For the connections with \cite{BonacinaGT14}, see Appendix \ref{app:k-winning}.

In Section \ref{sec:hall}, we present the proof of our version of the $(2-\epsilon)$-Hall's Lemma (Lemma \ref{lem:Hall}). This proof relies on a concentration result on the average right-degree and a discharging argument.  We also prove a bound for the best possible value  of $\epsilon$ for which Lemma \ref{lem:Hall} could hold and conjecture that this bound is in fact the optimal value of $\epsilon$ (Proposition \ref{thm:epsilon}).  

In Section \ref{sec:covergame}, we define a two player covering game \CovGam, whose aim is to dynamically build a \matching  inside a fixed bipartite graph $G$ (Definition \ref{def:cover-game}). 
Informally, a player, \Choose, queries nodes in the graph $G$ and the other player, \Cover, attempts to extend the current \matching to also cover the node queried (if not already covered).  
The main result of Section \ref{sec:covergame} is Theorem \ref{thm:covergame}, where we prove that if the graph $G$ has large left-expansion (i.e. large enough to apply Lemma \ref{lem:Hall} to sufficiently large subgraphs of $G$), then there is a winning strategy for \Cover to force \Choose to query a very large portion of the graph $G$.
In the analysis of the game, we use the $(2-\epsilon)$-Hall's Lemma and  \matchings in a similar manner to how matchings and $2$-matchings were used in \RES and \PCR\cite{Ben-SassonG03,Atserias04,BonacinaGT14,bg15}. 
A key difference is that  we are looking for winning strategies of \Cover for the \CovGam only on graphs $G$ where the number of high degree vertices is suitably bounded (Theorem \ref{thm:covergame}). This additional information allows us to identify a \matching covering all such high degree vertices in $G$ but preserving expansion properties of the remaining graph. \Cover will use this additional information to obtain a winning strategy. The full proofs of the technical Lemmas of this section are in Appendix \ref{app:lemmas}.

In Section \ref{sec:spacelowerbounds}, we prove (Lemma \ref{lem:cover-to-strategy}), that if \Cover wins \textsf{CoverGame} on the adjacency graph of a CNF $\phi$  (see Section 2 for the definition of adjacency graph) guaranteeing \matchings of maximal size $\mu$, then there exists a \WS{\mu} for the polynomial encoding of $\phi$. 
Finally, the monomial space in \PCR and the total space in \RES for random $3$-CNFs (Theorem \ref{cor:rand}) follow from well-known results about expansion of its adjacency graph \cite{ChvatalS88,BeameP96,Ben-SassonW01,Ben-SassonG03}. 
In order to get optimal lower bounds, we show in Lemma \ref{lem:big-degree} (with proof in Appendix \ref{app:high-degree}) that the number of variables appearing in many clauses of a random CNF is w.h.p. suitably bounded as required in the conditions of Theorem \ref{thm:covergame}. 

\section{Preliminaries}\label{sec:prelimin}

Let $X$ be a set of variables. A \dfn{literal} is a boolean constant, $0$ or $1$, or a variable $x\in X$, or the negation $\neg x$ of a variable $x$. 
A \dfn{clause} is a disjunction of literals: $C = (\ell_1 \vee \ldots \vee \ell_k)$. The \dfn{width} of a clause is the number of literals in it. 
A formula $\phi$ is in \dfn{Conjunctive Normal Form} (CNF) if $\phi = C_1 \wedge \ldots  \wedge C_m$ where $C_i$ are clauses. It is a $k$-CNF if each $C_i$ contains at most $k$ literals.
Let $\phi$ be a CNF and $X$ be the set of variables appearing in $\phi$. 
The \dfn{adjacency graph of $\phi$} is a bipartite graph $G_\phi$ with bipartition $(L,R)$ such that $L$ is the set of clauses of $\phi$,  $R=X$, and $(C,x)\in E$ if and only if $x$ or $\neg x$ appears in $C$.  If $\phi$ is a $k$-CNF, then $G_\phi$ has left-degree at most $k$.

\dfn{Resolution} (\RES) \cite{Blake37,Robinson:1965} is a propositional proof system for refuting unsatisfiable CNFs. 
Starting from an unsatisfiable CNF $\phi$, \RES allows us to derive the empty clause $\bot$ using the following inference rule:
$$
\frac{C\vee x \quad D\vee \neg x}{C\vee D}.
$$

\smallskip

Following \cite{AlekhnovichBRW02}, we define $\overline{X}= \{\bar{x}\ :\ x \in X\}$, which we regard as a set of formal variables with the intended meaning of $\bar x$ as $\neg x$. 
Given a field $\mathbb{F}$,  the ring $\mathbb{F}[X,\overline{X}]$ is the ring of polynomials in the variables $X\cup \overline{X}$ with coefficients in $\mathbb{F}$. 
We use the following \dfn{standard encoding ($tr$)} of CNF formulas over $X$ into a set of polynomials in $\mathbb{F}[X,\overline X] $: 
$tr(\phi)= \{tr(C)\ :\ C\in \phi\}\cup\{x^2-x, x+\bar x-1 \st x\in X\}$, where
$$
tr(x) = \bar x, \qquad  tr(\neg x) = x ,\qquad  tr( \bigvee_{i=1}^{n} \ell_i ) = \prod_{i=1}^n tr(\ell_i).
$$
A set of polynomials $P$ in $\mathbb{F}[X]$ is \dfn{contradictory} if and only if  $1$ is in the ideal generated by $P$. 
Notice that a CNF $\phi$ is unsatisfiable if and only if $tr(\phi)$ is a contradictory set of polynomials.

\dfn{Polynomial Calculus with Resolution} (\PCR) \cite{AlekhnovichBRW02} is an algebraic proof system for polynomials in $\mathbb{F}[X,\overline{X}]$. 
Starting from an initial set of contradictory polynomials $P$  in $\mathbb{F}[X,\overline{X}]$, \PCR allows us to derive the polynomial $1$ using the following inference rules: for all $p, q\in \mathbb{F}[X,\overline{X}]$
$$
\frac{p \quad \quad q }{\alpha p+\beta q}\ \forall\alpha,\beta \in \mathbb{F}, 
\qquad\qquad\qquad 
\frac{\quad p \quad}{vp} \forall v \in X\cup \overline{X}.
$$

To force 0/1 solutions, we always include the \dfn{boolean axioms} $\{x^2-x, x +\overline{x}-1\}_{x \in X}$ among the initial polynomials, as in the case of the polynomial encoding of CNFs. 

\smallskip

In order to study space of proofs we follow a model inspired by the definition of space complexity for Turing machines, where a machine is given a read-only input tape from which it can download parts of the input to the working memory as needed \cite{EstebanT01}. 

Given an unsatisfiable CNF formula $\phi$, a \dfn{\RES (resp. \PCR) refutation} of $\phi$ is a sequence  $\Pi=\langle\memory{0},\ldots,\memory{\ell}\rangle$ of sets of clauses (resp. polynomials), called \dfn{memory configurations}, such that: $\memory{0} =\emptyset$,  $\bot\in \memory{\ell}$ (resp. $1\in \memory{\ell}$), and for all $i\leq \ell$, $\memory{i}$ is obtained by $\memory{i-1}$ by applying one of the following rules:

\begin{quote}
\textsc{(Axiom Download)} $\memory{i} = \memory{i-1} \cup \{C\}$, where $C$ is a clause of $\phi$ (resp. a polynomial of $tr(\phi)$);\\
\textsc{(Inference Adding)} $\memory{i} = \memory{i-1} \cup \{O\}$, where $O$ is inferred by the \RES inference rule (resp. \PCR inference rules) from clauses (resp. polynomials) in \memory{i};\\
\textsc{(Erasure)} $\memory{i} \subset \memory{i-1}$.
\end{quote}

If in the definition of \PCR refutation we substitute the \textsc{Inference Adding} rule with:
\begin{quote}
 \textsc{(Semantical Inference)} $\memory{i}$ is contained in the  ideal generated by $\memory{i-1}$ in $\mathbb{F}[X,\overline X]$,
\end{quote}
we have what is called a \dfn{semantical \PCR refutation} of  $\phi$ \cite{AlekhnovichBRW02}.

The \dfn{total space} of $\Pi$ is the maximum over $i$ of the number of variables (counted with repetitions) occurring in \memory{i}.

The \dfn{monomial space} of a \PCR refutation $\Pi$, denoted by $\Space(\Pi)$, is the maximum over $i$ of the number of {\em distinct} monomials appearing in \memory{i}.

\subsection{Space lower bounds and \WSs{\boldsymbol{k}}}

A \dfn{partial assignment} over a set of variables $X$ is a map $\alpha:X\longrightarrow \{0,1,\star\}$. 
The \dfn{domain} of $\alpha$ is $\dom(\alpha)=\alpha^{-1}(\{0,1\})$. 
Given a partial assignment $\alpha$ and a CNF $\phi$ we can apply $\alpha$ to $\phi$, obtaining a new formula $\alpha(\phi)$ in the standard way, i.e. substituting each variable $x$ of $\phi$ in $\dom(\alpha)$ with the value $\alpha(x)$ and then simplifying the result. 
We say that $\alpha$ {\em satisfies} $\phi$, and we write  $\alpha\models \phi$, if $\alpha(\phi)=1$. 
Similarly, for a family $F$ of partial assignments, $F\models \phi$ means that for each $\alpha\in F$, $\alpha\models \phi$.

For each partial assignment $\alpha$ over $X \cup \overline{X}$ we assume that it respects the intended meaning of the variables; that is, $\alpha(\bar x)=1-\alpha(x)$ for each $x,\bar x \in \dom(\alpha)$. 
Given a partial assignment $\alpha$ and a polynomial $p$ in $\mathbb{F}[X,\overline X]$, we can apply $\alpha$ to $p$, obtaining a new polynomial $\alpha(p)$ in the standard way, similarly as before.
The notation $\alpha\models p$ means that $\alpha(p)=0$.
If $F$ is a family of partial assignments and $P$ a set of polynomials, we write $F\models P$ if $\alpha\models p$ for each $\alpha\in F$ and $p\in P$. 
Notice that if $\phi$ is a CNF and $\alpha$ is a partial assignment then $\alpha\models\phi$ if and only if $\alpha\models tr(\phi)$.

Let $A$ be a family of partial assignments, and let $\dom(A)$ be the union of the domains of the assignments in $A$. 
We say that a set of partial assignments $A$  is \dfn{flippable}  if and only if  for all $x \in \dom(A)$ there exist $\alpha,\beta \in A$ such that $\alpha(x)=1-\beta(x)$.
Two families of partial assignments $A$ and $A'$ are \dfn{domain-disjoint} if $\dom(\alpha)$ and $\dom(\alpha')$ are disjoint for all $\alpha\in A$ and $\alpha'\in A'$.
Given non-empty and pairwise domain-disjoint sets of assignments\footnote{We always suppose that the partial assignments are respecting the intended meaning of the variables in $\overline X$.  That is, if $x\in \dom(\alpha)$, then $\alpha(\bar x)=1-\alpha(x)$; hence a variable $x$ is in $\dom(H_i)$ if and only if $\bar x$ is in $\dom(H_i)$.} $\xvec{H}{1}{t}$, the \dfn{product-family} $\HH=H_1\otimes \ldots \otimes H_t$ is the following set of assignments
$$
\HH=H_1\otimes \ldots \otimes H_t= \{\alpha_1\cup\ldots \cup \alpha_t \st \alpha_i\in H_i\},
$$
or, if $t=0$, $\HH=\{\Star\}$, where $\Star$ is the partial assignment of the empty domain.
Note $\dom(\HH)=\bigcup_{i}\dom(H_i)$.  We call the $H_i$ the \dfn{factors} of \HH.
For a product-family $\HH= H_1\otimes \ldots \otimes H_t$, the \dfn{rank} of \HH, denoted $\|\HH\|$, is the number of factors of $\HH$ different from $\{\Star\}$. We do not count $\{\Star\}$ in the rank since $\HH\otimes \{\Star\}=\HH$. 
Given two product-families \HH and $\HH'$, we write $\HH'\sqsubseteq \HH$ if and only if each factor of $\HH'$ different from $\{\Star\}$ is also a factor of \HH. In particular, $\{\Star\}\sqsubseteq \HH$ for every \HH.

A family of \FPFs is called a \dfn{strategy} and denoted by \LL.  
We now present a definition of suitable families of flippable products: the {\em \WSs{k}} \cite{bg15}
.

\begin{definition}[\WS{k}\cite{bg15}]
\label{def:kex}
Let $P$ be a set of polynomials in the ring $\mathbb{F}[X,\overline X]$. A non-empty strategy \LL is a \dfn{\WS{\boldsymbol{k}}} if and only if for every $\HH \in  \LL$ the following conditions hold:
 \begin{description}
 \item[(restriction)] for each $\HH' \sqsubseteq \HH$, $\HH' \in \LL$;
 \item[(extension)] if  $\|\HH\|<k$, then for each $p \in P$  there exists a \FPF $\HH'\in \LL$ such that $\HH'\sqsupseteq \HH$ and $\HH' \models p$.
 \end{description}
\end{definition}
Notice that, by the restriction property, $\{\Star\}$ is in every \WS{k}.

\begin{theorem} \label{thm:lowerbound}
Let $\phi$ be an unsatisfiable CNF and $k\geq 1$ an integer.  If there exists a non-empty \WS{k}  \LL for $tr(\phi)$, then for every semantical \PCR refutation $\Pi$ of $\phi$, $\Space(\Pi) \geq \frac{k}{4}$. Moreover any resolution refutation of $\phi$ must pass through a memory
configuration containing at least $\frac{k-1}{2}$ clauses each of width at least $\frac{k-1}{2}$.
In particular, the \RES refutation requires total space at least $\frac{(k-1)^2}{4}$.
\end{theorem}

The monomial space lower bound follows directly from the main theorem of \cite{bg15}. In Appendix \ref{app:k-winning} we show how to use \WSs{k} to construct the combinatorial objects used in  \cite{BonacinaGT14} to obtain total space lower bounds.

\section{A $\boldsymbol{(2-\epsilon)}$-Hall's Lemma for \matchings}\label{sec:hall}

We now prove our variant of Hall's Lemma. 
This lemma may be of independent interest. 

\restatetheorem{\restateHall*}
 \noindent\textit{Proof.}
Observe that each vertex $v$ in $L$ has degree 2 or 3, otherwise $v$ could not be covered by a \matching.  Similarly, no degree 2 vertices in $L$ have the same neighbourhood. 
 
\begin{wrapfigure}[21]{l}{.26\textwidth}
\includegraphics[scale=.8]{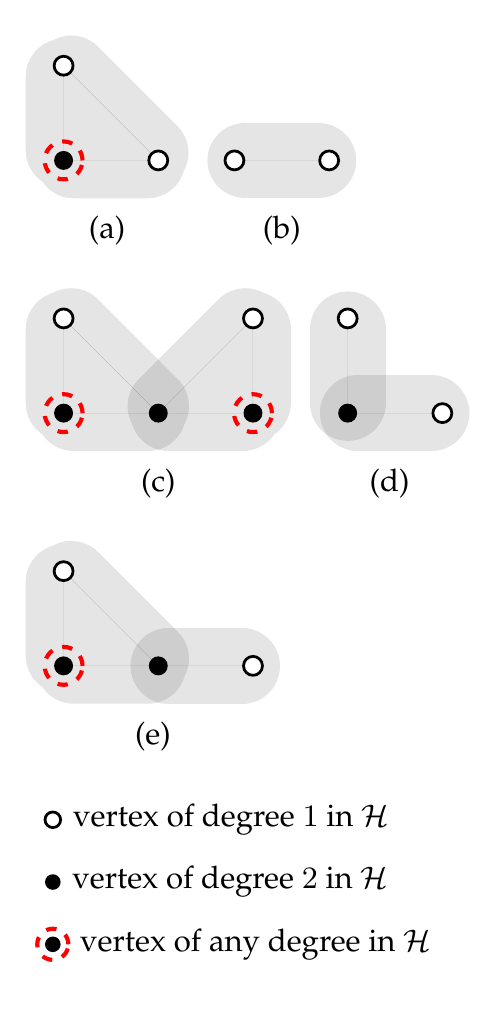}
\caption{A set of reducible configurations for $\mathcal{H}$}
\label{fig:reductions}
\end{wrapfigure}
By assumption, no degree 3 vertices have the same neighbourhood.  
Define the hypergraph $\calH=(V, E)$ with $V=N_G(L)$ and $E=\{N_G(x): x \in L\}$. 

By the above observations, $N_G: L \rightarrow E$ is a bijection and $|V|\geq (2-\epsilon)|L|=(2-\epsilon)|E|$.  The degree of a vertex $v$ in $\calH$, denoted $\deg_{\calH}(v)$, is the number of distinct hyperedges which contain $v$.
Let $L' \subseteq L$ and let $E' = \{N_G(x): x \in L'\}$.  
The existence of a \matching in $G$ covering $L'\subseteq L$ is equivalent to the existence of an injective function $f:E' \rightarrow \{\{x, y\}: x, y \in N_G(L')\}$, which we call a \dfn{$\boldsymbol{2}$-path cover} of $E'$,  such that
\begin{enumerate}
\item for every $e \in E'$,  $f(e)$ is a subset of size $2$ of $e$;
\item for each triple of distinct hyperedges $e_1,e_2,e_3\in E'$, it is not the case that $f(e_{i})\cap f(e_{i+1})\neq \emptyset$ for $i=1,2$.
\end{enumerate}

Observe that all the configurations of hyperedges shown in Figure \ref{fig:reductions} have a 2-path cover using only degree $1$ and $2$ vertices of \calH.
If any of these configurations appear in \calH, we can by assumption find a 2-path cover $f$ of the remaining hyperedges, and then extend $f$ to a 2-path cover of \calH.
Therefore, we may assume that no configuration from Figure \ref{fig:reductions} appears in \calH and we show that this assumption leads to a contradiction.

Let $d$ be the average degree in \calH. Observe that
$
3|E|\geq \sum_{v \in V} \deg_{\calH}(v)=d |V| \geq d(2-\epsilon)|E|,
$
where the last inequality follows from the hypothesis that $|V|\geq (2-\epsilon)|E|$.  Thus, $d \leq \frac{3}{2-\epsilon}$.

Let $D$ be the set of vertices of \calH of degree 1, and \calD be the set of hyperedges which contain a vertex in $D$.  
Note that $|D|=|\calD|$, since the configurations $(a)$, $(b)$ of Figure \ref{fig:reductions} do not appear in \calH.  We now prove a concentration result for $|D|$.
An upper bound follows immediately from the following chain of inequalities
$$
|D|=|\calD|\leq |E|\leq \frac{1}{2-\epsilon}|V|.
$$
For the lower bound suppose $\frac{1-2\epsilon}{2-\epsilon}|V| > |D|$.  Then,
\begin{align*}
\frac{3}{2-\epsilon} |V|  \geq d|V| = \sum_{v\in D }\deg_\calH(v) +\sum_{v\in V\setminus D}\deg_\calH(v) \\
\geq |D|+2|V\setminus D|> \frac{1-2\epsilon}{2-\epsilon} |V| + 2 (1-\frac{1-2\epsilon}{2-\epsilon}) |V|=\frac{3}{2-\epsilon} |V|,
\end{align*}
which is a contradiction. Hence we have that $\frac{1-2\epsilon}{2-\epsilon}|V| \leq |D| \leq \frac{1}{2-\epsilon}|V|$.

We finish the proof with a discharging argument.  Each vertex $v$ of \calH receives a charge of $\deg_{\calH}(v)$.  Let $E_2$ be the set of hyperedges in $E$ of size $2$ and $\calD_2$ be $\calD\cap E_2$.  The edges in 
$\calD_2$ receive a charge of -2, the edges in $\calD\setminus \calD_2$ receive charge $-3$.  The edges not in \calD receive no charge.

We now perform the following discharging rule.  Each
hyperedge $e$ in \calD gives a charge of -1 to each vertex in $e$.  After discharging, every hyperedge has charge 0, and every vertex has non-negative charge.
Let $Z$ be the set of vertices with charge 0 after discharging.  Observe that a vertex $x$ is in $Z$ if and only if every hyperedge containing $x$ also contains a degree 1 vertex.
 
Let $C$ denote the total charge.  Then,
\begin{align*}
C
=
 3|E|-|E_2|-3|\calD|+|\calD_2|
\leq
 3|E|-3|D| 
\leq
 \frac{3}{2-\epsilon} |V|-3\frac{1-2\epsilon}{2-\epsilon} |V|
=
\frac{6\epsilon}{2-\epsilon}|V|.
\end{align*}
It follows that $|Z| \geq \frac{2-7\epsilon}{2-\epsilon} |V|$.  Since $D \subseteq Z$ and $|D| \leq \frac{|V|}{2-\epsilon}$, we conclude that $|Z \setminus D| \geq \frac{1-7\epsilon}{2-\epsilon} |V|$.
Because  the  configurations $(c),(d),(e)$ of Figure \ref{fig:reductions} do not appear in \calH, every vertex in $Z \setminus D$ has degree at least 3.  Thus,
\begin{align*}
\frac{3}{2-\epsilon} |V|   &\geq d|V| \geq \sum_{v \in D} \deg_{\calH} (v)+ \sum_{v \in Z \setminus D} \deg_{\calH} (v) \geq |D| + 3|Z \setminus D| \\&\geq \frac{1-2\epsilon}{2-\epsilon}|V| + 3\frac{1-7\epsilon}{2-\epsilon} |V|  = \frac{4-23\epsilon}{2-\epsilon}|V|.
\end{align*}

Therefore, $3 \geq 4-23\epsilon$, which is a contradiction as $\epsilon < \frac{1}{23}$.\qed

\medskip
We end this section with a comment on the parameter $\epsilon$.  Since \matchings expand by a factor of at least $\frac{3}{2}$, we certainly require $\epsilon \leq \frac{1}{2}$ in the statement
of Lemma \ref{lem:Hall}.  In Proposition \ref{thm:epsilon}, we show a stronger upper bound that  $\epsilon \leq \frac{1}{3}$ is in fact necessary (see Appendix \ref{appendix:epsilon} for the proof).

\begin{proposition}\label{thm:epsilon}
For all $\epsilon>\frac{1}{3}$ there exists a bipartite graph $G_\epsilon$ with bipartition $(L,R)$ such that each vertex in $L$ has degree at most $3$ and no pair of degree $3$ vertices in $L$ have the same set of neighbours. Moreover, $|N_{G_\epsilon}(L)|\geq (2-\epsilon)|L|$ and each proper subset of $L$ can be covered by a \matching but $L$ cannot be covered by a \matching. 
\end{proposition}

\noindent We conjecture that Lemma \ref{lem:Hall} is true for $\epsilon \leq \frac{1}{3}$.  Proposition \ref{thm:epsilon} shows that this would be best possible.

\section{A cover game over bipartite graphs}
\label{sec:covergame}

As an application, we use the previous result to build a winning strategy for a game played on bipartite graphs.

\begin{definition}[Cover Game]  \label{def:cover-game}
The \dfn{Cover Game \CoverGame{G,\mu}} is a game between two players, 
\Choose and \Cover, on a bipartite graph $G$ with bipartition $(L,R)$. 
At each step $i$ of the game the players maintain a \matching $F_i$ in $G$. 
At step $i+1$ \Choose can 
\begin{enumerate}
\item remove a connected component from $F_i$, or
\item if the number of connected components of $F_i$ is strictly less than $\mu$, pick a vertex (either in $L$ or $R$) and challenge \Cover to find a \matching $F_{i+1}$ in $G$
such that
\begin{enumerate}
\item $F_{i+1}$ extends $F_i$.  That is, each connected component of $F_i$ is also a connected component of $F_{i+1}$;
\item $F_{i+1}$ covers the vertex picked by \Choose.
\end{enumerate}
\end{enumerate}
\Cover  loses the game \CoverGame{G,\mu} if at some point she cannot answer a challenge by \Choose.  Otherwise, \Cover wins.  
\end{definition}

\begin{definition}[\expander{s,\delta}] Let $s$ be a positive integer and $\delta$ be a positive real number. A bipartite graph  $G$ with bipartition $(L,R)$ is an \dfn{\expander{\boldsymbol{s,\delta}}} if all subsets $X \subseteq L$ of size at most $s$ satisfy $|N_{G}(X)| \geq \delta|X|$.
\end{definition}

The next theorem shows that \Cover has a winning strategy for the game \CoverGame{G, \mu} for expander graphs $G$ with appropriately chosen parameters.

\begin{theorem}\label{thm:covergame}
Let $G$ be a bipartite graph with bipartition $(L,R)$, $s,D$ be integers, and $\epsilon <\frac{1}{23}$ be a real number. 
For every integer $d \geq D$ let $S_d\subseteq R$ be the set of vertices of $R$ with degree bigger than $d$. Suppose that 
\begin{enumerate}
\item each vertex in $L$ has degree $3$;
\item $G$ is an \expander{s,2-\frac{\epsilon}{2}};
\item for every $d\geq D$, $\frac{72d}{\epsilon}(|S_d|+d)+1 \leq \frac{s}{2}$.
\end{enumerate}
Then \Cover wins the cover game \CoverGame{G,\mu} with 
$\mu=\frac{\epsilon s}{144 D}$.
\end{theorem}

The proof of this result is similar to constructions that can be found for example in \cite{Ben-SassonG03,Atserias04, BonacinaGT14}. 

For the rest of this section, fix a bipartite graph $G$ with bipartition $(L,R)$, an integer $s$ and a real number $\epsilon < \frac{1}{23}$ such that $G$ is an \expander{s,2-\frac{\epsilon}{2}} and each vertex in $L$ has degree $3$.
Given $A\subseteq L$ and $B\subseteq R$, we let $G_{A,B}$ be the subgraph of $G$ induced by $(L\cup R) \setminus (A\cup B)$.

\begin{definition}[\matchingprop]
Given two sets $A \subseteq L$ and $B \subseteq R$, we say that the pair $(A,B)$ has the \dfn{\matchingprop}, if for every $C \subseteq L \setminus A$ with
$|C| \leq s$, there exists a \matching $F$ in $G_{A,B}$ covering $C$.
\end{definition}

\begin{restatable}{\retheorem}{restateLemmaSmallC}\label{lem:smallC}
Let $A\subseteq L$ and $B\subseteq R$ be such that the pair $(A,B)$ does not have the \matchingprop.  Then there exists a set $C\subseteq L\setminus A$ with  $|C|< \frac{2}{\epsilon}|B|$, such that no \matching in $G_{A,B}$ covers $C$.
\end{restatable}

\begin{proof}
Take $C\subseteq L\setminus A$ of minimal size such that no \matching in $G_{A,B}$ covers $C$. We have that $|C|\leq s$ and by minimality of $C$ and Lemma \ref{lem:Hall} it follows that 
$$
|N_{G_{A,B}}(C)|< (2-\epsilon)|C|.
$$
But, by hypothesis $G$ is an \expander{s,2-\frac{\epsilon}{2}}; hence $(2-\frac{\epsilon}{2})|C|\leq |N_G(C)|$. Therefore,
$$
(2-\frac{\epsilon}{2})|C|\leq |N_G( C)|\leq |N_{G_{A,B}}(C)|+ |B|<(2-\epsilon)|C|+|B|.
$$
Hence $|C|<\frac{2}{\epsilon} |B|$, as required.
\end{proof}
Lemma \ref{lem:smallC} is the only place where we directly use the $(2-\epsilon)$-Hall's Lemma (Lemma \ref{lem:Hall}) from the previous section. However, Lemma \ref{lem:smallC} itself plays a crucial role in proving the following Lemmas (see Appendix \ref{app:lemmas} for the proofs).

\begin{restatable}{\retheorem}{restateLemmaEmpty} \label{lem:empty}
The pair $(\emptyset,\emptyset)$ has the \matchingprop.
\end{restatable}

\vspace{-1em}
\begin{restatable}[component removal]{\retheorem}{restateLemmaComponentRemoval}
\label{lem:comp-removal}
Let $A\subseteq L$ and $B\subseteq R$ be such that the pair $(A,B)$ has the \matchingprop and $\frac{2}{\epsilon}|B|\leq s$. 
Then for each \matching $F$ contained in the subgraph of $G$ induced by $A\cup B$,
$(A\setminus L(F),B\setminus R(F))$ has the \matchingprop.
\end{restatable}

\vspace{-1em}
\begin{restatable}[covering a vertex in $L$]{\retheorem}{restateLemmaCoveringL} \label{lem:left-vertex-cover}
Let $A\subseteq L$ and $B\subseteq R$ be such that the pair $(A,B)$ 
has the \matchingprop and let $d$ be the maximum degree of a vertex in $R\setminus B$.
If $\frac{24d}{\epsilon}(|B|+3)  + 1 \leq s$, then for each vertex $v$ in $L\setminus A$, there is a \matching $F$ in $G_{A,B}$ covering $v$ and such that $(A \cup L(F), B \cup R(F))$ has the \matchingprop.
\end{restatable}
\vspace{-1em}
\begin{restatable}[covering a vertex in $R$]{\retheorem}{restateLemmaCoveringR} \label{lem:right-vertex-cover}
Let $A\subseteq L$ and $B\subseteq R$ be such that the pair $(A,B)$ 
has the \matchingprop and let $d$ be the maximum degree of a vertex in $R\setminus B$.
If $\frac{24d}{\epsilon}(|B|+3d)  + 1 \leq s$, then for each vertex $v$ in $R\setminus B$, there is a \matching $F$ in $G_{A,B}$ covering $v$ and such that $(A \cup L(F), B \cup R(F))$ has the \matchingprop.
\end{restatable}


\begin{proof}[Proof of Theorem \ref{thm:covergame}]
 
By the hypothesis on $|S_d|$, for each $d\geq D$, we can repeatedly apply Lemma \ref{lem:right-vertex-cover} starting from $(\emptyset,\emptyset)$ to cover vertexes in $R$ of degree larger than $D$. By starting from vertices of $R$ of maximum degree and proceeding in decreasing order until reaching the vertices of degree $D$, we can build a \matching $M$ covering $S_D$ such that $(L(M),R(M))$ has the \matchingprop. 
Moreover, by the choice of $S_D$,  $G_{L(M),R(M)}$ (the subgraph induced by $(L\cup R)\setminus (L(M) \cup R(M))$) has degree at most $D$.
We say that a \matching $F$ is {\em compatible} with $M$ if each connected component of $F$ is either a connected component of $M$ or disjoint from all connected components of $M$.

We describe a winning strategy \LL for \Cover to win \CoverGame{G,\mu}.
Take \LL to be the set of all \matchings $F$ in $G$ compatible with $M$ such that 
\begin{enumerate} 
\item $(L(M)\cup L(F),R(M)\cup R(F))$ has the \matchingprop, and
\item $\frac{2}{\epsilon}|R(M)\cup R(F)|\leq s$.
\end{enumerate}

This family is non-empty since the empty \matching is in \LL.  Moreover, \LL is closed under removing connected components by Lemma \ref{lem:comp-removal}. Suppose now that at step $i+1$ of the game \Choose picks a vertex $v$ in $G_{L(M),R(M)}$ and that $F_i$ has strictly  less than $\mu=\frac{\epsilon s}{144 D}$ components.  
 Then, $(L(M)\cup L(F_i),R(M)\cup R(F_i))$ satisfies the hypotheses of Lemma \ref{lem:left-vertex-cover} and Lemma \ref{lem:right-vertex-cover}:
\begin{align*}\label{eq:boh}
\frac{24D}{\epsilon}(|R(M)\cup R(F_i)|+3D)  + 1 
&\leq 
\frac{24D}{\epsilon}(|R(M)|+|R(F_i)|+3D)  + 1
\\
&\leq 
\frac{24D}{\epsilon}(|R(M)|+3D)  + 1 + \frac{24D}{\epsilon}|R(F_i)|
\\
&\stackrel{(\star)}{\leq}
\frac{24D}{\epsilon}(3|S_D|+3D)  + 1 + \frac{72D}{\epsilon}\mu
\\
&\stackrel{(\star\star)}{\leq}
\frac{s}{2}+\frac{72D}{\epsilon}\mu
=
\frac{s}{2}+\frac{72D}{\epsilon}\frac{\epsilon s}{144 D}
=
s,
\end{align*}
where the inequality $(\star)$ follows from the fact that $|R(F_i)|\leq 3 \mu$ and $|R(M)|\leq 3|S_D|$, where $S_D$ is the set of vertices in $R$ of degree bigger than $D$. The inequality $(\star\star)$ follows by the hypothesis on the size of $S_D$.

Hence, if $v$ is covered by $F_i$ we take $F_{i+1}=F_i$. If $v$ is covered by $M$ we take $F_{i+1}=F_{i}\cup M_v$, where $M_v$ is the connected component of $M$ covering $v$. Otherwise, by  Lemma \ref{lem:left-vertex-cover} and Lemma \ref{lem:right-vertex-cover} applied to $(L(M)\cup L(F_i),R(M)\cup R(F_i))$, there exists a \matching $F_{i+1}$ extending $F_i\cup M$ by a new connected component covering $v$ such that 
$(L(F_{i+1}),R(F_{i+1}))$ has the \matchingprop.
From the previous chain of inequalities, it follows easily that the pair $(L(F_{i+1}),R(F_{i+1}))$ satisfies the cardinality condition $\frac{2}{\epsilon}|R(M)\cup R(F_{i+1})|=\frac{2}{\epsilon}|R(F_{i+1})|\leq s$.
\end{proof}

\section{Space lower bounds for random 3CNFs}\label{sec:spacelowerbounds}

\begin{lemma}
\label{lem:cover-to-strategy}
Let $\phi$ be an unsatisfiable $3$-CNF and $G_\phi$ its adjacency graph. If \Cover wins the cover game \CoverGame{G_\phi,\mu},  then there is a \WS{\mu} \LL for 
$tr(\phi)$.
\end{lemma}

\begin{wraptable}{L}{.3\textwidth}
\scriptsize
\footnotesize
\includegraphics[scale=.8]{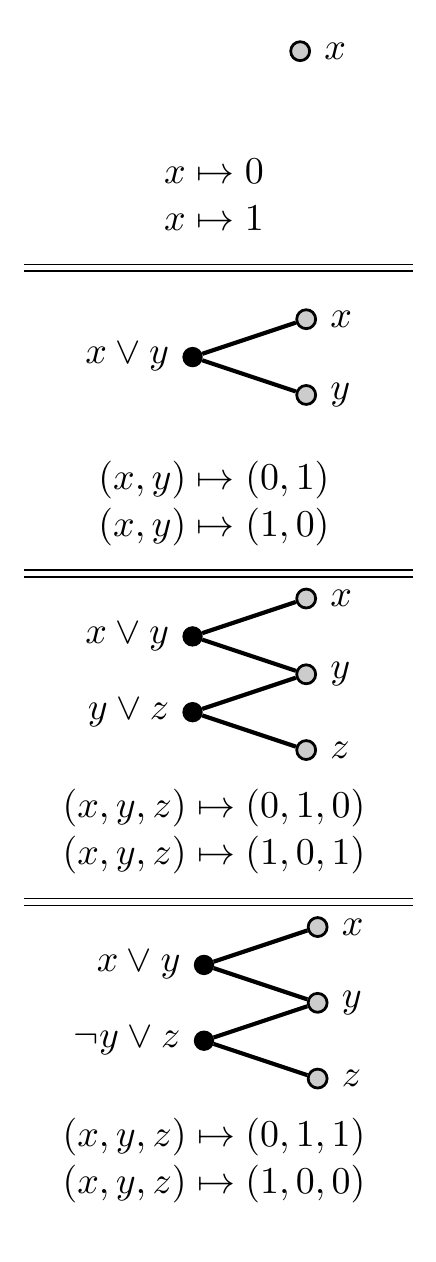}
\centering

\caption{Flippable assignments from \matchings}\label{table:flippable}
\end{wraptable}

\medskip
\noindent\textit{Proof.}
First of all we prove that for every \matching $F$ in $G_\phi$, there exists a flippable product-family of assignments $H_F$ such that $H_F\models L(F)$, $\dom(H_F)=R(F)$, and $\|H_F\|$ is the number of connected components of $F$.

We prove the result by induction on the number of connected components of $F$.
If $F$ is the union of two disjoint \matchings $F',F''$ then by hypothesis $H_{F'}\models L(F')$, $\dom(H_{F'})=R(F')$ and $\|H_{F'}\|$ is the number of connected components of $F'$. And analogously for $F''$. Then, since $R(F')$ and $R(F'')$ are disjoint, $H_F=H_{F'}\otimes H_{F''}$ is well-defined. We immediately see that $H_{F}\models L(F)$, $\dom(H_{F})=R(F)$ and $\|H_{F}\|$ is the number of connected components of $F$.

It remains to consider the case when the \matching $F$ is just one connected component. It is easy to see that all the possibilities can be reduced to those in Table \ref{table:flippable}.

It is straightforward to check that a winning strategy for \Cover in \CoverGame{G_\phi,\mu} defines, by previous observations,  a family \LL of flippable product-families such that for all $\HH\in \LL$
 \begin{enumerate}
 \item for each $\HH' \sqsubseteq \HH$, $\HH' \in \LL$;
 \item if  $\|\HH\|<\mu$, then:  (a) for each $C \in \phi$, there exists a \FPF $\HH'\in \LL$ such that $\HH'\models C$ and $\HH'\sqsupseteq \HH$; and (b) for each variable $x\not \in \dom(\HH)$, there exists a flippable family $\HH'\in \LL$ such that $\HH'\sqsupseteq \HH$ and $x \in \dom(\HH')$.
 \end{enumerate}

We claim that \LL is a \WS{\mu}. The {\em restriction property} is immediate. For the {\em extension property} we use the properties in (2) above: if we have to extend to something in \LL that satisfies a boolean axiom we use property 2.(b), otherwise for all other polynomials in $tr(\phi)$ we use property 2.(a). \qed 

Let $n,\Delta\in \mathbb{N}$ and let $X=\{x_1, \ldots , x_n\}$ be a set of $n$ variables.   The probability distribution $\R(n,\Delta,3)$ is obtained by the following experiment: choose independently uniformly at random $\Delta n$ clauses from the set of all possible clauses with $3$ literals over $X$.  
It is well-known that when $\Delta$  exceeds a certain constant $\theta_3$, $\phi$ is almost surely unsatisfiable \cite{ChvatalS88,BeameP96,Ben-SassonW01,Ben-SassonG03}.
Hence we always consider $\phi\sim \R(n,\Delta,3)$, where $\Delta$ is a constant bigger than $\theta_3$, which implies that $\phi$ is unsatisfiable with high probability.
The proof of the next Lemma is in Appendix \ref{app:high-degree}.

\begin{restatable}{\retheorem}{restateLemmaBigDegree} \label{lem:big-degree}
Let $\Delta > \theta_3$
and  $\phi\sim \R(n,\Delta,3)$ a random $3$-CNF. For every integer $d$ let $S_d$ be the set of variables of $\phi$ appearing in at least $d$ clauses of $\phi$. 
Then for every constant $c>0$ and $\epsilon > 0$, with high probability there exists a constant $D$ such that  for every $d\geq D$, \begin{equation*}
 \frac{72d}{\epsilon}(|S_d|+d)+1 \leq cn.
 \end{equation*}
\end{restatable} 

\begin{theorem} \label{cor:rand}\label{thm:rand}
If $\Delta > \theta_3$
and  $\phi\sim \R(n,\Delta,3)$, then the following statements hold with high probability.  For every semantical \PCR refutation $\Pi$ of $tr(\phi)$, $\Space(\Pi) \geq \Omega(n)$.
 Moreover, every \RES refutation of $\phi$ must pass through a memory
configuration containing $\Omega(n)$ clauses each of width $\Omega(n)$. In particular, each refutation of $\phi$ requires total space  $\Omega(n^2)$.
\end{theorem}
\begin{proof}
Let $G_\phi$ be the adjacency graph of $\phi$. It is well known that $G_\phi$ is a \expander{\gamma n, 2-\delta}, for every $\delta>0$ \cite{ChvatalS88,BeameP96,Ben-SassonW01,Ben-SassonG03}. 
Hence in particular for $0<\delta < \frac{1}{23}$ and using Lemma \ref{lem:big-degree} with $c=\frac{\gamma}{2}$, we satisfy all the hypotheses of Theorem \ref{thm:covergame}. Thus, \Cover wins the cover game $\CoverGame{G_\phi,\mu}$ for $\mu = \Omega(n)$.  Lemma \ref{lem:cover-to-strategy} provides a \WS{\Omega(n)} and by Theorem \ref{thm:lowerbound} we have the monomial space lower bound in semantical \PCR and the total space lower bound in \RES.
\end{proof}

\bibliographystyle{alpha}

\bibliography{merge}

\newcommand{\etalchar}[1]{$^{#1}$}
\begin{thebibliography}{ABSRW02}

\bibitem[ABSRW02]{AlekhnovichBRW02}
Michael Alekhnovich, Eli Ben-Sasson, Alexander~A. Razborov, and Avi Wigderson.
\newblock Space complexity in propositional calculus.
\newblock {\em SIAM J. Comput.}, 31(4):1184--1211, 2002.

\bibitem[AD08]{AtseriasD08}
Albert Atserias and V\'{\i}ctor Dalmau.
\newblock A combinatorial characterization of resolution width.
\newblock {\em J. Comput. Syst. Sci.}, 74(3):323--334, 2008.

\bibitem[Ats04]{Atserias04}
Albert Atserias.
\newblock On sufficient conditions for unsatisfiability of random formulas.
\newblock {\em J. ACM}, 51(2):281--311, 2004.

\bibitem[Ben02]{Ben-Sasson02}
Eli Ben{-}Sasson.
\newblock Size space tradeoffs for resolution.
\newblock In John~H. Reif, editor, {\em Proceedings on 34th Annual {ACM}
  Symposium on Theory of Computing, May 19-21, 2002, Montr{\'{e}}al,
  Qu{\'{e}}bec, Canada}, pages 457--464. {ACM}, 2002.

\bibitem[BG]{bg15}
Ilario Bonacina and Nicola Galesi.
\newblock A framework for space complexity in algebraic proof systems.
\newblock {\em J. ACM}, to appear.
\newblock Manuscript available at
  \url{http://wwwusers.di.uniroma1.it/~galesi/jacm.pdf}. A preliminary version
  appeared as: \emph{Pseudo-partitions, transversality and locality: a
  combinatorial characterization for the space measure in algebraic proof
  systems}. In ITCS, pages 455--472, 2013.

\bibitem[BGT14]{BonacinaGT14}
Ilario Bonacina, Nicola Galesi, and Neil Thapen.
\newblock Total space in resolution.
\newblock In {\em 55th Annual IEEE Symposium on Foundations of Computer Science
  (FOCS)}, pages 641--650, 2014.

\bibitem[Bla37]{Blake37}
Archie Blake.
\newblock {\em Canonical Expressions in Boolean Algebra}.
\newblock PhD thesis, 1937.
\newblock University of Chicago.

\bibitem[BN08]{Ben-SassonN08}
Eli Ben{-}Sasson and Jakob Nordstr{\"{o}}m.
\newblock Short proofs may be spacious: An optimal separation of space and
  length in resolution.
\newblock In {\em 49th Annual {IEEE} Symposium on Foundations of Computer
  Science, {FOCS} 2008, October 25-28, 2008, Philadelphia, PA, {USA}}, pages
  709--718. {IEEE} Computer Society, 2008.

\bibitem[BN11]{BenSassonN11}
Eli Ben{-}Sasson and Jakob Nordstr{\"{o}}m.
\newblock Understanding space in proof complexity: Separations and trade-offs
  via substitutions.
\newblock In Bernard Chazelle, editor, {\em Innovations in Computer Science -
  {ICS} 2010, Tsinghua University, Beijing, China, January 7-9, 2011.
  Proceedings}, pages 401--416. Tsinghua University Press, 2011.

\bibitem[BP96]{BeameP96}
Paul Beame and Toniann Pitassi.
\newblock Simplified and improved resolution lower bounds.
\newblock In {\em FOCS}, pages 274--282. IEEE Computer Society, 1996.

\bibitem[BSG03]{Ben-SassonG03}
Eli Ben-Sasson and Nicola Galesi.
\newblock Space complexity of random formulae in resolution.
\newblock {\em Random Struct. Algorithms}, 23(1):92--109, 2003.

\bibitem[BSW01]{Ben-SassonW01}
Eli Ben-Sasson and Avi Wigderson.
\newblock Short proofs are narrow - resolution made simple.
\newblock {\em J. ACM}, 48(2):149--169, 2001.

\bibitem[CEI96]{CleggEI96}
Matthew Clegg, Jeff Edmonds, and Russell Impagliazzo.
\newblock Using the groebner basis algorithm to find proofs of
  unsatisfiability.
\newblock In Gary~L. Miller, editor, {\em STOC}, pages 174--183. ACM, 1996.

\bibitem[CS88]{ChvatalS88}
Vasek Chv{\'a}tal and Endre Szemer{\'e}di.
\newblock Many hard examples for resolution.
\newblock {\em J. ACM}, 35(4):759--768, 1988.

\bibitem[ET01]{EstebanT01}
Juan~Luis Esteban and Jacobo Tor{\'a}n.
\newblock Space bounds for resolution.
\newblock {\em Inf. Comput.}, 171(1):84--97, 2001.

\bibitem[FLM{\etalchar{+}}13]{FilmusLMNV13}
Yuval Filmus, Massimo Lauria, Mladen Mik{\v s}a, Jakob Nordstr{\"o}m, and Marc
  Vinyals.
\newblock Towards an understanding of polynomial calculus: New separations and
  lower bounds - (extended abstract).
\newblock In Fedor~V. Fomin, Rusins Freivalds, Marta~Z. Kwiatkowska, and David
  Peleg, editors, {\em ICALP (1)}, volume 7965 of {\em Lecture Notes in
  Computer Science}, pages 437--448. Springer, 2013.

\bibitem[FLN{\etalchar{+}}12]{FilmusLNTR12}
Yuval Filmus, Massimo Lauria, Jakob Nordstr{\"{o}}m, Neil Thapen, and Noga
  Ron{-}Zewi.
\newblock Space complexity in polynomial calculus.
\newblock In {\em Proceedings of the 27th Conference on Computational
  Complexity, {CCC} 2012, Porto, Portugal, June 26-29, 2012}, pages 334--344.
  {IEEE}, 2012.

\bibitem[Hal35]{Hall}
P.~Hall.
\newblock On representatives of subsets.
\newblock {\em Journal of the London Mathematical Society}, s1-10(1):26--30,
  1935.

\bibitem[Nor09]{Nordstrom09}
Jakob Nordstr{\"o}m.
\newblock Narrow proofs may be spacious: Separating space and width in
  resolution.
\newblock {\em SIAM J. Comput.}, 39(1):59--121, 2009.

\bibitem[Nor13]{Nordstrom13}
Jakob Nordstr{\"{o}}m.
\newblock Pebble games, proof complexity, and time-space trade-offs.
\newblock {\em Logical Methods in Computer Science}, 9(3), 2013.

\bibitem[Rob65]{Robinson:1965}
J.~A. Robinson.
\newblock A machine-oriented logic based on the resolution principle.
\newblock {\em J. ACM}, 12(1):23--41, January 1965.

\end{thebibliography}

\appendix

\section{Proof of Theorem \ref{thm:lowerbound}}\label{app:k-winning}

A \dfn{piecewise (p.w.) assignment} $\alpha$ of a set of variables $X$ is a set of non-empty partial assignments to $X$ with pairwise disjoint domains.
We will sometimes call the elements of $\alpha$ the \emph{pieces} of $\alpha$. 
A piecewise assignment gives rise to  a partial assignment $\bigcup \alpha$ to $X$ together
with a partition of the domain of $\bigcup \alpha$. 
For piecewise assignments $\alpha, \beta$ we will write $\alpha \sqsubseteq \beta$ to mean 
that every piece of $\alpha$ appears in $\beta$. We will write $\| \alpha \|$ to mean the number
of pieces in~$\alpha$. Note that these are formally exactly the same as $\alpha \subseteq \beta$
and $|\alpha|$, if we regard $\alpha$ and $\beta$ as sets.

\begin{definition}[\free{r}\cite{BonacinaGT14}]\label{def:r-free}
A family $\FF$ of p.w. assignments is 
\emph{\free{r} for a CNF $\phi$} if it has the following properties:
\begin{description}
\item[(Consistency)]
No $\alpha \in \FF$ falsifies any clause from $\phi$;
\item[(Retraction)]
If $\alpha \in \FF$, $\beta$ is a p.w. assignment and $\alpha^\star\sqsubseteq \beta \sqsubseteq \alpha$, then $\beta \in \FF$;
\item[(Extension)]
If $\alpha \in \FF$ and $\|\alpha\| < r$, then for every variable 
$x \notin \dom (\alpha)$, there exist $\beta_0, \beta_1 \in \FF$ 
with $\alpha \sqsubseteq \beta_0, \beta_1$ such that 
$\beta_0(x)=0$ and $\beta_1(x)=1$.
\end{description}
\end{definition}

\begin{theorem}[\cite{BonacinaGT14}]\label{thm:r-free}
Let $\phi$ be an unsatisfiable CNF formula.
If there is a family of p.w. assignments which is \free{r} for $\phi$,
then any resolution refutation of $\phi$ must pass through a memory
configuration containing at least $\frac{r}{2}$ clauses each of width at least $\frac{r}{2}$. In particular, the refutation requires total space at least $\frac{r^2}{4}$.
\end{theorem}

By this theorem, in order to prove the total space lower bound of Theorem \ref{thm:lowerbound} we just have to prove that given a \WS{k} for $tr(\phi)$ we can build a \free{(k-1)} family for $\phi$.

\medskip
Let \LL be the \WS{k}.
Define the  $(k-1)$-free  family \FF as follows:
$\alpha \in \FF$ if and only if  there exists $H_1\otimes \ldots \otimes H_t \in \LL$ such that $\alpha = \alpha_1\cup\ldots \cup \alpha_t$
with $\alpha_i\in H_i$ and $t \leq k-1$. 
The p.w. structure of $\alpha$ is inherited from the domain-disjointness of $H_1\otimes \ldots \otimes H_t$; in particular, $\|\alpha\|=\|H_1\otimes \ldots \otimes H_t\|$.
The {\em retraction property} of \FF is immediate from the corresponding property of \LL.

To prove the {\em consistency property} of \FF assume, by contradiction, that there is an $\alpha \in \FF$ such that $\alpha$ falsifies some clause $C \in \phi$. 
Since $||\alpha||\leq k-1<k$,  there exists $\HH= H_1\otimes \ldots \otimes H_t \in \LL$ such that $\alpha\in \HH$ and $\|\alpha\|=\|\HH\|$. By the extension property of \LL, there is an $\HH'\sqsupseteq \HH$ such that $\HH'\models tr(C)$.  In particular there exists some partial assignment $\beta\supseteq \alpha$ such that $\beta\models tr(C)$.  By construction, for every assignment $\gamma$, $\gamma\models tr(C)$ if and only if $\gamma\models C$.  Thus $\beta\models C$, which is impossible since $\alpha$ falsifies $C$.

For  the {\em extension property} let $\alpha \in \FF$, with $||\alpha||<k-1$ and let $x$ be a variable of $\phi$ not in $\dom(\alpha)$. By construction, there exists some $\HH\in \LL$ such that $\alpha\in \HH$, $\|\alpha\|=\|\HH\|$ and $\dom(\alpha)=\dom(\HH)$. By the extension property of \FF there exists some flippable $\HH'\in \LL$ such that $\HH'\sqsupseteq \HH$ and $\HH'\models x^2-x$. By taking restrictions in \LL we can suppose that $\|\HH'\|=\|\HH\|+1$.
Hence there exist $\beta_0,\beta_1\in \FF$ extending $\alpha$, setting $x$ respectively to $0$ and $1$ and such that  $||\beta_0||=||\beta_1||=\|\alpha\|+1\leq k-1$. \qed

\section{Proofs of the Lemmas from Section \ref{sec:covergame}}\label{app:lemmas}

For  convenience we restate here also Lemma \ref{lem:smallC}.

\restatetheorem{\restateLemmaSmallC*}

\restatetheorem{\restateLemmaEmpty*}

\begin{proof}
Apply Lemma \ref{lem:smallC} with $A=\emptyset$ and $B=\emptyset$.
\end{proof}

\restatetheorem{\restateLemmaComponentRemoval*}

\begin{proof}
Let $A'=A\setminus L(F)$ and $B'=B\setminus R(F)$ and suppose, by contradiction, that $(A',B')$ does  not have the \matchingprop. 
By Lemma \ref{lem:smallC}, it is sufficient to prove that for each set $C\subseteq L\setminus A'$ with $|C|<\frac{2}{\epsilon}|B'|$, there is a \matching in $G_{A',B'}$ covering $C$. Let $C'=C\cap L(F)$ and $C''=C\setminus C'$. By construction, $F$ is a \matching such that $L(F)\subseteq A$, $R(F)\subseteq B$ and $F$ covers $C'$. Moreover, we have that
 $$
 |C''|\leq|C|<\frac{2}{\epsilon}|B'|<\frac{2}{\epsilon}|B|\stackrel{(\star)}{\leq} s,
 $$
 where the inequality $(\star)$ is by hypothesis. Hence there exists a \matching $F''$ of $C''$ in $G_{A,B}$, and so $F \cup F''$ is a \matching covering $C$ in $G_{A',B'}$.
\end{proof}

\restatetheorem{\restateLemmaCoveringL*}

\begin{proof}
Fix $v\in L\setminus A$ and let $\Pi$ be the set of all \matchings $F$ in $G_{A,B}$, covering $v$ and such that $F$ is connected.

Since $1 \leq s$ and $(A,B)$ has the \matchingprop, we know that
$\Pi$ is non-empty. 
For every $F\in \Pi$, let $(A_F,B_F)$ be the pair $(A \cup L(F), B \cup R(F))$, and suppose for a contradiction that for every $F \in \Pi$, $(A_F, B_F)$ does not have the \matchingprop. 
By Lemma \ref{lem:smallC}, for every $F \in \Pi$ there is a set
$C_F \subseteq L \setminus A_F$
with $|C_F| < \frac{2}{\epsilon}|B_F|$ and such that there is no \matching 
of $C_F$ in $G_{A_F,B_F}$.

Let $C = \bigcup_{F \in \Pi} C_F$.
Then 
\begin{equation*}
|C| \leq \sum_{F\in \Pi}|C_F|<|\Pi|\frac{2}{\epsilon}(|B|+3)\leq 12 d\frac{2}{\epsilon}(|B|+3),
\end{equation*}
since $|\Pi| \leq 3+ 3\cdot 2\cdot (d-1) \cdot 2\leq 12 d$ and $|B_F|\leq |B|+ 3$. Hence,
by our assumption about the size of $|B|$, we have that $|C \cup \{v \}| \leq s$.
Furthermore, $C \cup \{ v \} \subseteq L \setminus A$, so by the fact that $(A,B)$ has the \matchingprop, there is a \matching $F'$ covering $C \cup \{ v \}$ in $G_{A,B}$.

There must be some $F \in \Pi$ such that $F$ is a connected component of $F'$. Let $F''$ be $F'$
with the component $F$ removed. Then $F''$ 
is a \matching in $G_{A_F,B_F}$ and $F''$ covers $C_F$, contradicting the choice of $C_F$. 
\end{proof}

\restatetheorem{\restateLemmaCoveringR*}

\begin{proof}
Fix $v \in R\setminus B$ and let $D$ be $N_G(v)\setminus A$. By hypothesis $|D|\leq d$. If $|D|=0$, then $N_G(v)\subseteq A$, and so we can cover $v$ by taking $F$ to be
the \matching consisting only of the vertex $v$. This is a valid \matching covering $v$ and clearly $(A \cup L(F), B \cup R(F))$ has the \matchingprop.

If $|D|>0$, by the cardinality condition on $B$, we can apply Lemma \ref{lem:left-vertex-cover} $|D|$ times obtaining a \matching $F$ in $G_{A,B}$ covering $D$ and such that $(A \cup L(F), B \cup R(F))$ has the \matchingprop.

Now, since $N_G(v)\subseteq A\cup L(F)$, it follows that $(A\cup L(F),B\cup R(F)\cup \{v\})$ has the \matchingprop. Either $v$ is covered by $F$, or it is possible to add $\{v\}$ as a new connected component to $F$ while still maintaining the property of being a \matching in $G_{A,B}$.
\end{proof}

\section{Proof of Proposition \ref{thm:epsilon}}\label{appendix:epsilon}

As promised, we now prove Proposition \ref{thm:epsilon}, here rephrased in terms of hypergraphs.  

\begin{figure}[h]
\centering
\includegraphics{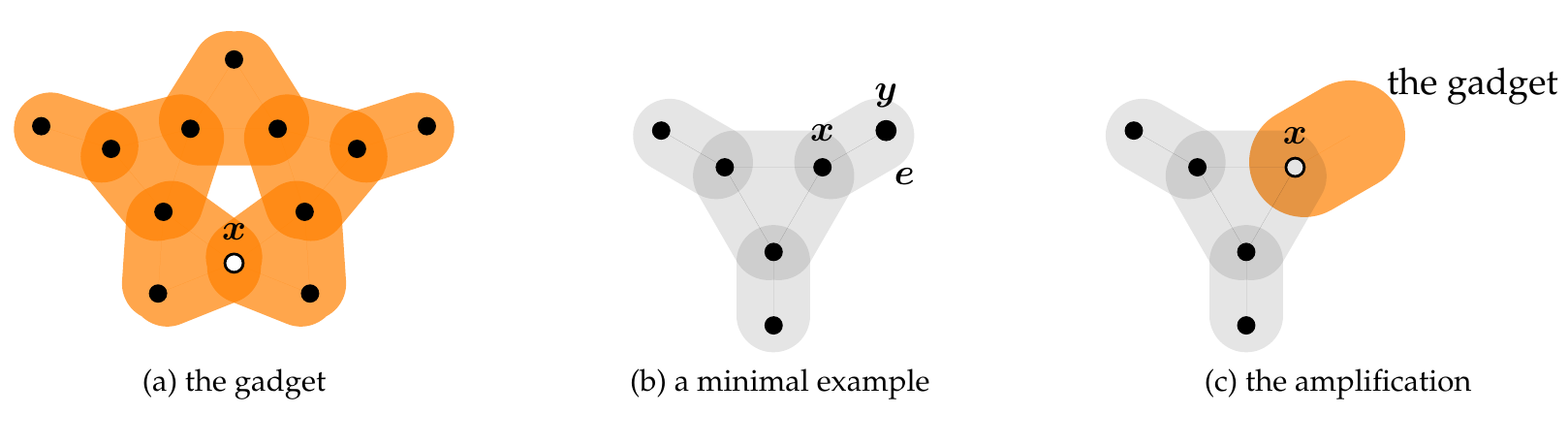}
\caption{The construction}
\label{fig:ampl}
\end{figure}

\begin{proposition}For every $\epsilon>\frac{1}{3}$, there exists a hypergraph $\mathcal{H}_\epsilon$ such that $\mathcal{H}_\epsilon$ has no isolated vertices, each hyperedge of $\mathcal{H}_\epsilon$ has size 2 or 3, $|V(\mathcal{H})|\geq (2-\epsilon)|E(\mathcal{H})|$, every proper subset of $E(\mathcal{H}_{\epsilon})$ has a 2-path cover, but $\mathcal{H}_{\epsilon}$ does not have a 2-path cover. 
\end{proposition}

\begin{proof}
Let $\epsilon>\frac{1}{3}$ and consider the gadget $\mathcal{G}$ shown in Figure~\ref{fig:ampl}.(a).  It is easy to verify that every 2-path cover of $\mathcal{G}$ must cover the vertex $x$.  
Next note that the hypergraph $\mathcal{H}$ shown in Figure~\ref{fig:ampl}.(b) is obviously not 2-path coverable, but every proper subset of $E(\mathcal{H})$ is
2-path coverable.  We have $\frac{|V(\mathcal{H})|}{|E(\mathcal{H})|}=\frac{6}{4}$.  However, we can increase this ratio via the amplification trick shown in 
Figure~\ref{fig:ampl}.(c).  

That is, let $e$ be a hyperedge of $\mathcal{H}$ of size 2.  Label the vertices of $e$ as $x$ and $y$, where $y$ has degree 1.  
Let $\mathcal{H}_1$ be the hypergraph obtained from $\mathcal{H}$ by deleting $y$ and then gluing $\mathcal{G}$ to  $\mathcal{H}-y$ along $x$.  Since every
2-path cover of $\mathcal{G}$ must use the vertex $x$, $\mathcal{H}_1$  does not have a 2-path cover.  On the other hand, since every proper subset
of $E(\mathcal{G})$ has a 2-path cover avoiding $x$, it follows that every proper subset of $E(\mathcal{H}_1)$ has a 2-path cover.  Note that this amplification 
trick increases the number of vertices of $\mathcal{H}$ by 10 and the number of edges of  $\mathcal{H}$ by 6.  
Moreover, we can repeat this amplification trick arbitrarily many times since $\mathcal{G}$ also has pendent edges of size 2.  So, choose $n$ such that $\frac{6+10n}{4+6n} \geq 2-\epsilon$ and
take  $\mathcal{H}_\epsilon$ to be the graph obtained from $\mathcal{H}$ by performing the amplification trick $n$ times.  
\end{proof}

\section{Proof of Lemma \ref{lem:big-degree}}\label{app:high-degree}

\restatetheorem{\restateLemmaBigDegree*}

\begin{proof}
Let $G_\phi$ be the adjacency graph of $\phi$. First of all we show that w.h.p. there are at most $\frac{e n }{2^d}$ many variable nodes of degree $d$ for every $d\geq 24e\Delta$ and that w.h.p. there is no variable node of degree bigger than $\log n$.
First note that the expected number of variable nodes of degree at least $\log n$ is 

\begin{equation*}
n \binom{\Delta n}{\log n} \bfrac{3}{n-2}^{\log n} 
\le
n \bfrac{e\Delta n}{\log n}^{\log n} \bfrac{3}{n-2}^{\log n} 
= 
o(1).
\end{equation*}
So w.h.p. there are no such nodes.
Let $d \ge 24e\Delta$. The probability that there are $\frac{e n }{2^d}$ many variable nodes of degree $d$ is at most 
\begin{equation*}
\binom{n}{\frac{e n }{2^d}} \sqbs{\binom{\Delta n}{d} \bfrac{3}{n-2}^d}^{\frac{e n }{2^d}} 
\leq 
\bfrac{en}{\frac{en }{2^d}}^\frac{e n }{2^d}\sqbs{\bfrac{e\Delta n}{d}^d \bfrac{3}{n-2}^d}^{\frac{e n }{2^d}}
\leq 
\bfrac{12e\Delta}{d}^{\frac{ed n }{2^d}} 
\leq 
\bfrac{1}{2}^{\frac{ed n }{2^d}},
\end{equation*}
so, by the union bound, the probability that there exists any $d$ between $24e\Delta$ and $\log n$ such that there are $\frac{e n }{2^d}$ many variable nodes of degree $d$ is at most $ \sum_{24e\Delta \leq d \leq \log n} \bfrac{1}{2}^{\frac{ed n }{2^d}}$. 
To bound this sum, note that the ratio of consecutive terms is 
\begin{equation*}
2^{\frac{edn}{2^d} - \frac{e(d+1) n}{2^{d+1}}} = 2^{\frac{e(d-1)n}{2^{d+1}}} \ge 2
\end{equation*}
 for $d$ in this range, and so the sum is of the order of its last term, which is $\bfrac{1}{2}^{\frac{e n \log n}{2^{\log n}}} = o(1)$. 

So we have that w.h.p.
\begin{equation*}
|S_d|\leq \sum_{d'\geq d}\frac{en}{2^{d'}}\leq \frac{2en}{2^d}
\end{equation*}
and, for (a not yet chosen constant) $D\geq 24e\Delta$, we have that for each $d$ such that $D\leq d\leq \log n$:
\begin{eqnarray*}
 \frac{72d}{\epsilon}(|S_d|+d)+1 
\leq
 \frac{72d}{\epsilon}\left(\frac{2en}{2^d}+d\right)+1 
\leq 
\frac{72D}{\epsilon}\frac{2en}{2^D}+O(\log^2 n)
\leq
cn,
\end{eqnarray*}
where the last inequality holds if $D$ is a sufficiently large constant.  
\end{proof}

\end{document}